\documentclass[]{interact}

\usepackage{epstopdf}
\usepackage[caption=false]{subfig}
\usepackage{algorithm}
\usepackage{algpseudocode}
\newtheorem{assum}{Assumption}
\newtheorem{prop}{Proposition}
\newtheorem{defn}{Definition}

\newtheorem{rem}{Remark}

\usepackage[longnamesfirst,sort]{natbib}
\bibpunct[, ]{(}{)}{;}{a}{,}{,}
\renewcommand\bibfont{\fontsize{10}{12}\selectfont}

\theoremstyle{plain}

\theoremstyle{definition}

\theoremstyle{remark}

\begin{document}


\title{Stabilization with Closed-loop DOA Enlargement: An Interval Analysis Approach}

\author{
\name{X. Qiu\textsuperscript{a} Z. Feng\textsuperscript{a} C. Lu\textsuperscript{a} and Y. Li\textsuperscript{a}\thanks{CONTACT Y. Li. Email: yqli@zjut.edu.cn}}
}

\maketitle

\begin{abstract}
\looseness=-1 In this paper, the stabilization problem with closed-loop domain of attraction (DOA) enlargement for discrete-time general nonlinear plants is solved. First, a sufficient condition for asymptotic stabilization and estimation of the closed-loop DOA is given. It shows that, for a given Lyapunov function, the negative-definite and invariant set in the state-control space is a stabilizing controller set and its projection along the control space to the state space can be an estimate of the closed-loop DOA. Then, an algorithm is proposed to approximate the negative-definite and invariant set for the given Lyapunov function, in which an interval analysis algorithm is used to find an inner approximation of sets as precise as desired. Finally, a solvable optimization problem is formulated to enlarge the estimate of the closed-loop DOA by selecting an appropriate Lyapunov function from a positive-definite function set. The proposed method try to find a unstructured controller set (namely, the negative-definite and invariant set) in the state-control space rather than design parameters of a structured controller in traditional synthesis methods.
\end{abstract}


\section{Introduction}

Stabilization is one of the fundamental problems in the control science. For nonlinear plants, due to the difficulty to achieve the global stabilization, the domain of attraction (DOA) of the closed-loop system requires extensive investigation. 

DOA is an invariant set characterizing asymptotically stabilizable area around the equilibrium, from which all state trajectories emanating converge to the equilibrium \citep{Chesi:2011,Zarei:2018_164,Zheng:2018}. It is well known that DOA plays an important role both in analysis and synthesis. In analysis, considering the autonomous systems, several efforts have been made in estimating the DOA. As summarized in \cite{Zheng:2018}, existing methods involve maximal Lyapunov functions \citep{Vannelli:1985_69,Zarei:2018_164}, compositions of Lyapunov functions \citep{Tan:2008_565}, invariance principle approach \citep{Han:2016_5569}, trajectory reverse approach \citep{Genesio:1985_747},  occupation measure approach \citep{Henrion:2014_297}, \textit{etc.}. In synthesis, a few efforts have been made in simultaneously designing the controller and enlarging the estimate of the DOA of the closed-loop. For the state-quadratic and input-linear-saturating plant, \citet{Valmorbida:2010_1196} derives a LMI-based optimization problem for computing the state feedback gains maximizing the estimate of the closed-loop DOA. For the polynomial nonlinear plant with affine inputs, \citet{Chesi:2012_3322} designs a polynomial feedback controller enlarging the estimate of the closed-loop DOA by solving a generalized eigenvalue problem. For the rational polynomial nonlinear plant, \cite{Vatani:2017_4266} proposes a sum-of-squares programming based optimization problem to design a controller enlarging the estimate of the closed-loop DOA. For the nonlinear plant with affine inputs, \cite{Saleme:2013_4074} develops an optimization strategy based on a multidimensional gridding approach and \cite{Ding:2019_2545} constructs a second-order sliding-mode control algorithm, both which provide the enlarged estimate of the closed-loop DOA. For the general nonlinear plant, \cite{Davo:2018_73} presents a switching logic combining two output feedback controllers (one renders the closed-loop locally stable and the other provides ultimate boundedness with large DOA) to enlarge the closed-loop DOA. It is assumed that the two predesigned stabilizing controllers with bounded DOA are known. For the discrete-time general nonlinear plant, \cite{Li:2014_79} proposes a method to find a stabilization controller set and an enlarged estimate of the closed-loop DOA, which is based on the random sampling and set griding technique.

However, aforementioned works about synthesis, except \cite{Davo:2018_73} and \cite{Li:2014_79}, only consider specific nonlinear plants rather than general nonlinear plants, such as polynomial nonlinearity, rational polynomial nonlinearity or affine nonlinearity. The drawback of \cite{Davo:2018_73} is that the two predesigned stabilizing controllers with bounded DOA are required, while there is no systematic method designing such controllers. The drawbacks of \cite{Li:2014_79} include: 1) due to the random sampling and set griding technique approximating sets, there is no quantitative analysis result about set estimation errors; 2) the invariant of the estimate of DOA is guaranteed by the level-set of the Lyapunov function, which in general leads to a conservative estimate.

In this paper, considering discrete-time general nonlinear plants, the stabilization with closed-loop DOA enlargement is solved. First, a sufficient condition for asymptotic stabilization and estimation of the closed-loop DOA is given, which describes the principle of Lyapunov method from a new point of view. It shows that, for a given Lyapunov function, the negative-definite and invariant set in the state-control space is a stabilizing controller set (namely, any controller belonging to this set can asymptotically stabilize the plant) and its projection along the control space to the state space can be an estimate of the closed-loop DOA. Then, an algorithm is proposed to approximate the negative-definite and invariant set for the given Lyapunov function, in which the Set Inversion Via Interval Analysis (SIVIA) algorithm is used to find an inner approximation of sets. The SIVIA algorithm is one of the basic tools in interval analysis approach, which is a kind of numerical method to approximate sets of interest as precise as desired \citep{Jaulin:2001}. Finally, a solvable optimization problem is formulated to enlarge the estimate of the closed-loop DOA by selecting an appropriate Lyapunov function from a positive-definite function set.

For the discrete-time general nonlinear stabilization problem with the closed-loop DOA enlargement, the contributions of this paper include:

\begin{itemize}
\item A new synthesis method for discrete-time nonlinear systems is proposed, which utilize the principle of Lyapunov method from a new point of view. Its central concept is the negative-definite and invariant set in the state-control space. Rather than design parameters of a structured controller in traditional synthesis methods, our method try to find a unstructured controller set in the state-control space, namely, the negative-definite and invariant set.
	
\item The SIVIA algorithm, one of basic tools in interval analysis, is introduced to obtain an inner approximation of the negative-definite and invariant set, which can guarantee the convergence of the estimation of sets and the desired approximation precision.
\end{itemize}

The rest of this paper is organized as follows. In Section 2, a sufficient condition for asymptotic stabilization and estimation of the closed-loop DOA is given, in which the negative-definite and invariant set plays an important role. In Section 3, based on theoretical result in the preceding section, the stabilization problem with the closed-loop DOA enlargement is solve by using interval analysis approach to estimate the negative-definite and invariant set. In Section 4, the simulation result is presented. The conclusion is drawn in Section 5.

\textbf{Notation: } For a vector $x \in \mathbb{R}^n$, $x_{(i)}$ represents the $i$-th element of $x$, $i = 1,2,\cdots,n$. For two vectors $x \in \mathbb{R}^n$ and $u \in \mathbb{R}^m$, $w = (x;u)$ represents a new vector in $\mathbb{R}^{n+m}$. $[w] \subset \mathbb{R}^{n+m}$ represents a box belonging to $\mathbb{R}^{n+m}$ (\textit{e.g., }a rectangular region when $n+m = 2$). $\mathbb{W} \subset \mathbb{R}^{n+m}$ represents an arbitrary compact subset of $\mathbb{R}^{n+m}$. $\hat{\mathbb{W}}$ represents an approximation of $\mathbb{W}$ by covering $\mathbb{W}$ with non-overlapping boxes in a set of boxes (there is no ambiguity when $\hat{\mathbb{W}}$ is viewed as a set of boxes or the unions of boxes according to contents). For a set $\mathbb{W} \subset \mathbb{R}^{n+m}$, $\mathrm{proj}(\mathbb{W}) \subset \mathbb{R}^n$ denotes the orthogonal projection of $\mathbb{W}$ along $\mathbb{R}^m$ to  $\mathbb{R}^n$.

\section{Negative-definite and invariant set in state-control space}

Consider the nonlinear discrete-time system
\begin{equation} \label{eq:plant}
x(k+1)=f\left(x(k),u(k)\right), \; k = 0,1,2,\cdots,
\end{equation}
where $x(k) \in \mathbb{R}^n$ is the state, $u(k) \in \mathbb{R}^m$ is the control input and $f: \mathbb{R}^n \times \mathbb{R}^m \to \mathbb{R}^n$ is a continuous function in both arguments satisfying $0 = f(0,0)$ and that its linearization at the origin is controllable. 

Our control objective is to find a nonlinear feedback controller $\mu: \mathbb{R}^n \to \mathbb{R}^m$ such that the closed-loop $x(k+1)=$ $f(x(k),\mu(x(k)))$ is asymptotically stable at $x=0$ and that the estimate of the closed-loop DOA is as large as possible.

In this section, a sufficient condition for asymptotic stabilization and estimation of the closed-loop DOA is present, which describes the principle of Lyapunov method from a new point of view.

\subsection{Negative-definite set in state-control space}

For plant \eqref{eq:plant}, we omit time instant $k$ and let $x_+, x$ and $u$ denote $x(k+1), x(k)$ and $u(k)$, respectively. Then the dynamic of plant \eqref{eq:plant} can be represented as a hyper-surface $\Pi$ in $(2n+m)$-dimensional space, which is defined as
\begin{equation}
\Pi = \Big\{(x_+;x;u) \in \mathbb{R}^{2n+m} \Big| x_+ = f(x,u)\Big\}.
\end{equation}
For a given positive-definite function $L: \mathbb{R}^n \to \mathbb{R}$, which satisfies $L(0) = 0$ and $L(x) > 0, \forall x \in \mathbb{R}/\{0\}$, a subset $\Pi_{\mathrm{N}}(L)$ of $\Pi$ can be defined as
\begin{equation}
\Pi_{\mathrm{N}}(L) = \Big\{(x_+;x;u) \in \Pi \Big| L(x_+) - L(x) < 0\Big\}.
\end{equation}

Noting that $x_+$ denotes the future state at the next time instant and $x$ denotes the current state, it is obvious that any point $(x_+;x;u) \in \Pi_{\mathrm{N}}(L)$ makes the difference of the positive-definite function $L(x(k))$ is negative-definite. Hence, the subset $\Pi_{\mathrm{N}}(L) \subset \mathbb{R}^{2n+m}$ of the hyper-surface $\Pi$ is called the negative-definite set of plant \eqref{eq:plant} for $L$. For purpose of designing controller, projecting $\Pi_{\mathrm{N}}(L)$ along the future state space onto the state-control space, we get the negative-definite set $\mathbb{W}_{\mathrm{N}}(L)$ in the state-control space, which is defined as follows.

\begin{defn} \label{defn:negative-definite set}
	A subset $\mathbb{W}_{\mathrm{N}}(L)$ of the state-control space is said to be a negative-definite set of plant \eqref{eq:plant} for the positive-definite function $L: \mathbb{R}^n \to \mathbb{R}$ if
	\begin{eqnarray}
	\mathbb{W}_{\mathrm{N}}(L) = \Big\{(x;u) \in \mathbb{R}^{n+m} \Big|L(f(x,u)) - L(x) < 0\Big\}. \label{eq:defn_ndd}
	\end{eqnarray}
\end{defn}

In general, the negative-definite set $\mathbb{W}_{\mathrm{N}}(L)$ is unbounded and open. The unboundedness is obvious from \eqref{eq:defn_ndd}. The openness is due to that the boundary $\{(x;u)|L(f(x,u)) - L(x) = 0\}$ of $\mathbb{W}_{\mathrm{N}}(L)$ is not a subset of $\mathbb{W}_{\mathrm{N}}(L)$.

From Definition~\ref{defn:negative-definite set}, it is straightforward that the difference of $L(x(k))$ is negative-definite at $k$ if $(x(k);u(k)) \in \mathbb{W}_{\mathrm{N}}(L), \forall (x(k);u(k)) \neq 0 \in \mathbb{R}^{n+m}$. Based on this, one may conclude that the closed-loop is asymptotically stable for all initial states in $\mathrm{proj}(\mathbb{W}_{\mathrm{N}}(L)) \subset \mathbb{R}^n$ if the controller $\mu$ satisfies $\forall x \in \mathrm{proj}(\mathbb{W}_{\mathrm{N}}(L)), (x;\mu(x)) \in \mathbb{W}_{\mathrm{N}}(L)$, where $\mathrm{proj}(\mathbb{W}_{\mathrm{N}}) \subset \mathbb{R}^n$ represents the projection of $\mathbb{W}_{\mathrm{N}}(L)$ along the control space to the state space. Unfortunately, this conclusion is wrong. Because it can not be guaranteed that the state is still in $\mathrm{proj}(\mathbb{W}_{\mathrm{N}}(L))$ at $k+1$. Once the state is outside of $\mathrm{proj}(\mathbb{W}_{\mathrm{N}}(L))$, the condition that the difference of $L(x(k))$ is negative-definite is no longer satisfied. This problem can be solved if an invariant subset of $\mathbb{W}_{\mathrm{N}}(L)$ could be found, in which any point $(x;u)$ can guarantee that the future state $f(x,u)$ is still in $\mathrm{proj}(\mathbb{W}_{\mathrm{N}}(L))$. This gives rise to the definition of the invariant set in the state-control space.

\subsection{Invariant set in state-control space}

For an autonomous system $f_a: \mathbb{R}^n \to \mathbb{R}^n$, a set $\mathbb{X}_{\mathrm{I}} \subset \mathbb{R}^n$ is defined to be invariant \citep{Blanchini:1999_1747} if $\forall x \in \mathbb{X}_{\mathrm{I}}, f_a(x) \in \mathbb{X}_{\mathrm{I}}$. For plant \eqref{eq:plant}, a set $\mathbb{X}_{\mathrm{I}} \subset \mathbb{R}^n$ is defined to be control invariant \citep{Blanchini:1999_1747} if $\forall x \in \mathbb{X}_{\mathrm{I}}$, there exists a control input $u \in \mathbb{R}^m$ such that $f(x,u) \in \mathbb{X}_{\mathrm{I}}$. In literatures, both the invariant set and the control invariant set are defined in the state space and well researched \citep{Bertsekas:1972_604,Blanchini:1999_1747,Rakovic:2006_546}. Here, for purpose of designing controller, we define the invariant set in the state-control space as the following.

\begin{defn} \label{defn:invaraint_set}
	A set $\mathbb{W}_{\mathrm{I}} \subset \mathbb{R}^{n+m}$ is said to be invariant for plant \eqref{eq:plant} if
	\begin{equation}
	\mathbb{W}_{\mathrm{I}} = \Big\{(x;u) \in \mathbb{R}^{n+m}\Big|f(x,u) \in \mathrm{proj}(\mathbb{W}_{\mathrm{I}})\Big\}, \label{eq:defn:invariant_set}
	\end{equation}
	where $\mathrm{proj}(\mathbb{W}_{\mathrm{I}})$ represents the projection of $\mathbb{W}_{\mathrm{I}}$ along the control space to the state space.
\end{defn}

In order to find an invariant subset $\mathbb{W}_{\mathrm{I}}$ of the interested region in the state-control space, we define a new mapping between subsets of $\mathbb{R}^{n+m}$ as follows.

\begin{defn}
	For $\mathbb{W} \subset \mathbb{R}^{n+m}$, mapping $\mathcal{I}$ is defined as
	\begin{equation}
	\mathcal{I}(\mathbb{W}) = \Big\{w \in \mathbb{W}\Big|f(w) \in \mathrm{proj}(\mathbb{W})\Big\}. \label{eq:defn:I}
	\end{equation}
	where $\mathrm{proj}(\mathbb{W})$ represents the orthogonal projection of $\mathbb{W}$ along the control space to the state space.
\end{defn}

\begin{rem}
	Mapping $\mathcal{I}$ is similar with the mapping $\mathcal{C}$ defined in (11) of \cite{Bertsekas:1972_604}, which is defined for the plants with disturbance. For $\mathbb{W} \subset \mathbb{R}^{n+m}$, its non-disturbance version is 
	\begin{equation}
	\mathcal{C}(\mathrm{proj}(\mathbb{W})) = \Big\{w \in \mathbb{W}\Big|f(w) \in \mathrm{proj}(\mathbb{W})\Big\}. \label{eq:defn:C}
	\end{equation}
	From \eqref{eq:defn:I} and \eqref{eq:defn:C}, it is obvious that the fundamental principles of mappings $\mathcal{I}$ and $\mathcal{C}$ are same. The only difference is that $\mathcal{I}$ maps subsets of the state-control space into subsets of the state-control space, while $\mathcal{C}$ maps subsets of the state space into subsets of the state-control space.
\end{rem}

The composition of mapping $\mathcal{I}$ with itself $i$ times is denoted by $\mathcal{I}^i$. Mapping $\mathcal{I}$ has the following property.

\begin{prop} \label{prop:I_infty}
	For any compact set $\mathbb{W} \subset \mathbb{R}^{n+m}$, set limit $\mathcal{I}^{\infty}(\mathbb{W}) = \lim_{i \to \infty} \mathcal{I}^{i}(\mathbb{W}) = \bigcap_{i = 1}^{\infty} \mathcal{I}^i(\mathbb{W})$ exists and is invariant for plant~\eqref{eq:plant}.
\end{prop}

\begin{proof}
	From the definition of mapping $\mathcal{I}$, we have
	\begin{eqnarray}
	\mathcal{I}^{i+1}(\mathbb{W}) = \Big\{w \in \mathcal{I}^{i}(\mathbb{W})\Big|f(w) \in \mathrm{proj}(\mathcal{I}^{i}(\mathbb{W}))\Big\} \label{eq:prop:I_infty:pf:I_i}
	\end{eqnarray}
	According to \eqref{eq:prop:I_infty:pf:I_i}, it is obvious that  $\{\mathcal{I}^i (\mathbb{W})\}$ is a monotonically decreasing sequence of sets in the sense that $\mathcal{I}^{i+1} (\mathbb{W}) \subset \mathcal{I}^i (\mathbb{W})$ for all $i \geq 1$. Moreover, $\mathcal{I}^i (\mathbb{W})$ is a compact set for all $i \geq 1$. Hence, the set limit of $\mathcal{I}^i (\mathbb{W})$ exists (see \citealp{Rockafellar:2009}, page 111).
	
	Again, from \eqref{eq:prop:I_infty:pf:I_i}, we have
	\begin{eqnarray}
	\forall w \in \mathcal{I}^{i+1}(\mathbb{W}), f(w) \in \mathrm{proj}(\mathcal{I}^{i}(\mathbb{W})). \nonumber
	\end{eqnarray}
	When $i$ tends to infinite, because the set limit of $\mathcal{I}^i (\mathbb{W})$ exists, it follows that
	\begin{eqnarray}
	\forall w \in \mathcal{I}^{\infty}(\mathbb{W}), f(w) \in \mathrm{proj}(\mathcal{I}^{\infty}(\mathbb{W})),\nonumber
	\end{eqnarray}
	which means that set $\mathcal{I}^{\infty}(\mathbb{W})$ satisfies Definition~\ref{defn:invaraint_set}. Thus $\mathcal{I}^{\infty}(\mathbb{W})$ is invariant for plant~\eqref{eq:plant}. 
\end{proof}

From Proposition~\ref{prop:I_infty}, we know that if the mapping $\mathcal{I}$ is applied to any initial compact subset $\mathbb{W}$ of the state-control space infinite times, then an invariant subset of $\mathbb{W}$ can be obtained. Hence, it is possible to use the mapping $\mathcal{I}$ to find an invariant subset of the negative-definite set $\mathbb{W}_{\mathrm{N}}(L)$ defined in Definition~\ref{defn:negative-definite set}.

\subsection{Negative-definite and invariant set in state-control space}

From Proposition~\ref{prop:I_infty}, an invariant subset of a compact set in the state-control space can be found using the mapping $\mathcal{I}$. However, the negative-definite set $\mathbb{W}_{\mathrm{N}}(L)$ defined in Definition~\ref{defn:negative-definite set} is unbounded and open, namely $\mathbb{W}_{\mathrm{N}}(L)$ is not compact. In order to guarantee the boundedness of $\mathbb{W}_{\mathrm{N}}(L)$, we introduce the following assumption.

\begin{assum} \label{assum:cons_set}
	The state and control input satisfy a set of mixed constrains
	\begin{equation*}
	(x;u) \in \mathbb{W}_\mathrm{cons} \subset \mathbb{R}^{n+m},
	\end{equation*}
	where $\mathbb{W}_\mathrm{cons}$ is a compact set.
\end{assum}

It is without loss of generality to introduce Assumption~\ref{assum:cons_set}, because these constrains typically arise due to physical limitations or safety considerations in practice. Under Assumption~\ref{assum:cons_set}, it follows that $\mathbb{W}_{\mathrm{N}}(L) \subset \mathbb{W}_\mathrm{cons}$, therefore $\mathbb{W}_{\mathrm{N}}(L)$ is bounded. The openness of $\mathbb{W}_{\mathrm{N}}(L)$ is due to that its boundary $\{(x;u)|L(f(x,u)) - L(x) = 0\}$ is not its subset, therefore we modify \eqref{eq:defn_ndd} as 
\begin{equation} 
\mathbb{W}_{\mathrm{N}}(L) = \Big\{(x;u) \in \mathbb{R}^{n+m} \Big|L(f(x,u)) - L(x) \leq -\alpha \Big\}.\label{eq:ndd}
\end{equation}
where $\alpha \in \mathbb{R}_+$ is a very small positive constant. Then, it is obvious that, under Assumption~\ref{assum:cons_set}, $\mathbb{W}_{\mathrm{N}}(L)$ defined in \eqref{eq:ndd} is bounded and closed. Hence, $\mathbb{W}_{\mathrm{N}}(L)$ defined in \eqref{eq:ndd} is compact. With the compact set $\mathbb{W}_{\mathrm{N}}(L)$, we give the following proposition.

\begin{prop}\label{prop:W_N&I}
	Under Assumption~\ref{assum:cons_set}, the subset $\mathbb{W}_{\mathrm{N\&I}}(L)$ of the state-control space defined as
	\begin{equation} 
	\mathbb{W}_{\mathrm{N\&I}}(L) = \mathcal{I}^\infty(\mathbb{W}_{\mathrm{N}}(L)) \label{eq:W_N&I}
	\end{equation} 
	is a negative-definite and invariant set for plant \eqref{eq:plant} and positive-definite function $L: \mathbb{R}^n \to \mathbb{R}$, where the negative-definite set $\mathbb{W}_{\mathrm{N}}(L)$ is defined in \eqref{eq:ndd}.
\end{prop}

\begin{proof}
	Under Assumption~\ref{assum:cons_set}, the negative-definite set $\mathbb{W}_{\mathrm{N}}(L)$ defined in \eqref{eq:ndd} is compact. Following Proposition~\ref{prop:I_infty}, we have that the set limit $\mathcal{I}^\infty(\mathbb{W}_{\mathrm{N}}(L)) = \lim_{i \to \infty} \mathcal{I}^{i}(\mathbb{W}_{\mathrm{N}}(L))$ exists and is invariant for plant~\eqref{eq:plant}.
	
	From the definition of the mapping $\mathcal{I}$, we know that
	\begin{equation*}
	\mathcal{I}^\infty(\mathbb{W}_{\mathrm{N}}(L))\subseteq \cdots \subseteq \mathcal{I}^2(\mathbb{W}_{\mathrm{N}}(L)) \subseteq \mathcal{I}(\mathbb{W}_{\mathrm{N}}(L)) \subseteq \mathbb{W}_{\mathrm{N}}(L).
	\end{equation*}
	The above relations mean that $\mathbb{W}_{\mathrm{N\&I}}(L)$ is a subset of $\mathbb{W}_{\mathrm{N}}(L)$. Hence, $\mathbb{W}_{\mathrm{N\&I}}(L)$ is negative-definite for plant \eqref{eq:plant} and positive-definite function $L$.
\end{proof}

\subsection{Sufficient condition for stabilization and estimation of closed-loop DOA}

Because $\mathbb{W}_{\mathrm{N\&I}} (L)$ is negative-definite, from \eqref{eq:ndd}, $\forall (x(k);u(k)) \in \mathbb{W}_{\mathrm{N\&I}}(L)$, the difference of $L(x(k))$ is negative-definite at $k$. Because $\mathbb{W}_{\mathrm{N\&I}} (L)$ is invariant, from Definition~\ref{defn:invaraint_set}, $\forall (x(k);u(k)) \in \mathbb{W}_{\mathrm{N\&I}}(L)$, the future state $x(k+1)$ is in $\mathrm{proj}(\mathbb{W}_{\mathrm{N\&I}}(L))$ at $k+1$. This means that, for $x(k+1)$, there exists $u(k+1)$ such that the difference of $L(x(k+1))$ is also negative-definite at $k+1$. Hence, we can conclude that the closed-loop system is asymptotically stable for all initial states in $\mathrm{proj}(\mathbb{W}_{\mathrm{N\&I}}(L)) \subset \mathbb{R}^n$ if the controller $\mu$ satisfies $\forall x \in \mathrm{proj}(\mathbb{W}_{\mathrm{N\&I}}(L)), (x;\mu(x)) \in \mathbb{W}_{\mathrm{N\&I}}(L)$. This idea is summarized in Proposition~\ref{prop:stabilization}.

\begin{prop} \label{prop:stabilization}
	If set $\mathbb{W}_{\mathrm{N\&I}} (L) \subseteq \mathbb{R}^{n+m}$ is negative-definite and invariant for plant~\eqref{eq:plant} and Lyapunov function $L: \mathbb{R}^n \to \overline{\mathbb{R}}_+$, then, for any controller $\mu: \mathbb{R}^n \to \mathbb{R}^m$ satisfying
	\begin{equation}
	0 = \mu(0), (x; \mu(x)) \in \mathbb{W}_{\mathrm{N\&I}}(L), \forall x \in \mathrm{proj} (\mathbb{W}_{\mathrm{\mathrm{N\&I}}} (L)), \label{eq:prop:stab:mu_subset_W}
	\end{equation}
	the closed-loop system $x(k+1) = f\big(x(k),\mu(x(k))\big)$ is asymptotically stable for any initial state in $\mathrm{proj} (\mathbb{W}_{\mathrm{\mathrm{N\&I}}} (L))$.
\end{prop}

\begin{proof}
	Because $\mathbb{W}_{\mathrm{N\&I}} (L)$ is negative-definite for plant~\eqref{eq:plant} and Lyapunov function $L$, from \eqref{eq:ndd} and \eqref{eq:prop:stab:mu_subset_W}, it follows that
	\begin{equation}
	\forall x \in \mathrm{proj} (\mathbb{W}_{\mathrm{\mathrm{N\&I}}} (L)), L\big(f\left(x,\mu(x)\right)\big) - L(x) < 0. \label{eq:prop:stab:pf:negative-definite}
	\end{equation}
	Because $\mathbb{W}_{\mathrm{N\&I}} (L)$ is invariant for plant~\eqref{eq:plant}, from definition~\ref{defn:invaraint_set} and \eqref{eq:prop:stab:mu_subset_W}, it follows that
	\begin{equation}
	\forall x \in \mathrm{proj} (\mathbb{W}_{\mathrm{\mathrm{N\&I}}} (L)), f\left(x,\mu(x)\right) \in \mathrm{proj} (\mathbb{W}_{\mathrm{\mathrm{N\&I}}} (L)). \label{eq:prop:stab:pf:invariant}
	\end{equation}
	For plant~\eqref{eq:plant}, let $\phi(x_0,k)$ denote the solution of $x(k+1) = f\big(x(k),\mu(x(k))\big)$ at time $k$ with the initial state $x_0$. From \eqref{eq:prop:stab:pf:negative-definite} and \eqref{eq:prop:stab:pf:invariant}, it is obvious that
	\begin{equation}
	\forall x_0 \in \mathrm{proj} (\mathbb{W}_{\mathrm{\mathrm{N\&I}}} (L)), \forall k, L(\phi(x_0,k+1)) < L(\phi(x_0,k)). \label{eq:prop:stab:pf:L(k+1)<L(k)} \nonumber
	\end{equation}
	The above relation shows that, $\forall x_0 \in \mathrm{proj} (\mathbb{W}_{\mathrm{\mathrm{N\&I}}} (L))$, $L(\phi(x_0,k))$ is monotonically decreasing with time. And because $L$ is positive-definite, $L(\phi(x_0,k))$ is bounded from below by zero. Hence, 
	\begin{equation}
	\forall x_0 \in \mathrm{proj} (\mathbb{W}_{\mathrm{\mathrm{N\&I}}} (L)), \lim_{k \to \infty} L(\phi(x_0,k)) = 0. \nonumber
	\end{equation}
	From the above equation, we can derive that
	\begin{equation}
	\forall x_0 \in \mathrm{proj} (\mathbb{W}_{\mathrm{\mathrm{N\&I}}} (L)), \lim_{k \to \infty} \phi(x_0,k) = 0. \nonumber
	\end{equation}
	This can be proven by reductio ad absurdum (details see the proof of Theorem 13.2 in \citealp{Haddad:2008}). 
\end{proof}

From Proposition~\ref{prop:stabilization}, we know that, if the negative-definite and invariant set $\mathbb{W}_{\mathrm{N\&I}} (L) \subset \mathbb{R}^{n+m}$ can be obtained, any controller $\mu$ satisfying $(x;\mu(x)) \in \mathbb{W}_{\mathrm{N\&I}} (L)$ can asymptotically stabilize plant \eqref{eq:plant} and $\mathrm{proj}(\mathbb{W}_{\mathrm{N\&I}}(L)) \subset \mathbb{R}^n$ can be an estimate of the closed-loop DOA. However, due to nonlinearities of $f$ and $L$, it is hard to obtain analytic solution of $\mathbb{W}_{\mathrm{N\&I}} (L)$. In the next section, we use the interval analysis approach to approximate the negative-definite and invariant set $\mathbb{W}_{\mathrm{N\&I}} (L)$.

\section{Stabilization with closed-loop DOA enlargement: Interval analysis approach}

Based on the preceding section, now we use interval analysis approach to solve the stabilization problem with the closed-loop DOA enlargement. Firstly, interval analysis approach, especially the SIVIA algorithm, is briefly introduced. Then, an algorithm, estimating the negative-definite and invariant set $\mathbb{W}_{\mathrm{N\&I}}(L)$ for the given positive-definite function $L$ based on the SIVIA algorithm, is proposed. Thirdly, with the estimate of $\mathbb{W}_{\mathrm{N\&I}}(L)$, the closed-loop DOA is estimated and the controller is designed. Finally, the estimate of the closed-loop DOA is enlarged by selecting an appropriate Lyapunov function from a positive-definite function set.

\subsection{Interval analysis}

Interval analysis is a kind of guaranteed numerical method for approximating sets. Guaranteed means here that approximations of sets of interest are obtained, which can be made as precise as desired \citep{Jaulin:2001}. Algorithm SIVIA is one of basic tools in interval analysis, which uses a covering with non-overlapping boxes to approximate set inversion $\mathbb{Z} \subset \mathbb{R}^{n_1}$ defined by function $p: \mathbb{R}^{n_1} \to \mathbb{R}^{n_2}$ and set $\mathbb{Y} \subset \mathbb{R}^{n_2}$ as
\begin{equation}
\mathbb{Z} = \Big\{z \in \mathbb{R}^{n_1} \Big| p(z) \in \mathbb{Y} \Big\}. \label{eq:problem_SIVIA}
\end{equation}

The fundamental of SIVIA is the concept of interval vectors and inclusion functions. We briefly introduce these concepts (more details see \citealt{Jaulin:2001}). An interval vector $[z]$ is a subset of $\mathbb{R}^{n_1}$, which is defined as $[z] = [z_{(1)}] \times [z_{(2)}] \times \cdots \times [z_{(n_1)}]$, where the $j$-th interval $[z_{(j)}] = [\underline{z}_{(j)},\bar{z}_{(j)}], j = 1,2,\cdots,n_1$, is a connected subset of $\mathbb{R}$, $\underline{z}_{(j)}$ and $\bar{z}_{(j)}$ are the lower and the upper bound of the interval $[z_{(j)}]$. $[z] \in \mathbb{IR}^{n_1}$ is also called a box, where $\mathbb{IR}^{n_1}$ denotes the set of all $n_1$-dimensional boxes. Considering function $p: \mathbb{R}^{n_1} \to \mathbb{R}^{n_2}$, the interval function $[p]: \mathbb{IR}^{n_1} \to \mathbb{IR}^{n_2}$ is an inclusion function for $p$ if $\forall [z] \in \mathbb{IR}^{n_1}, p([z]) \subset [p]([z])$. An inclusion function $[p]$ is convergent if $\forall [z] \in \mathbb{IR}^{n_1}, \lim_{d([z]) \to 0}d([p]([z])) = 0$, where $d([z]) = \max_{1 \leq j \leq n_1} (\bar{z}_{(j)} - \underline{z}_{(j)})$ denotes the width of box $[z]$. The convergent inclusion function can guarantee the convergence of SIVIA algorithm. For a function, its convergent inclusion function is not unique, \textit{e.g.}, natural form, centered form and Taylor form.

SIVIA can find an inner approximation $\hat{\mathbb{Z}}_{\mathrm{in}} \subset \mathbb{Z}_{\mathrm{init}} \subset \mathbb{R}^{n_1}$ of $\mathbb{Z}$, where $\mathbb{Z}_{\mathrm{init}}$ is a given initial search set. SIVIA performs a recursive exploration (the while loop in Algorithm~\ref{alg:Sivia}), in which four cases may be encountered for a given box $[z] \subset \mathbb{R}^{n_1}$.

\begin{itemize}
	\item Inner test: if $[p]([z])$ is entirely in $\mathbb{Y}$, then $[z]$ is entirely in $\mathbb{Z}$, and is stored in set $\hat{\mathbb{Z}}_{\mathrm{in}}$ collecting boxes inside $\mathbb{Z}$, as shown in Line 7-8 in Algorithm~\ref{alg:Sivia}.
	\item Outer test: if $[p]([z])$ has an empty intersection with $\mathbb{Y}$, then $[z]$ does not belong to $\mathbb{Z}$, and is stored in set $\hat{\mathbb{Z}}_{\mathrm{out}}$ collecting boxes outside $\mathbb{Z}$, as shown in Line 9-10 in Algorithm~\ref{alg:Sivia}.
	\item If $[p]([z])$ has a non-empty intersection with $\mathbb{Y}$, but is not entirely in $\mathbb{Y}$, then $[z]$ contains the boundary of $\mathbb{Z}$; $[z]$ is said to be undetermined. If the width of $[z]$ is lower than a prespecified parameter $\epsilon > 0$, then it is deemed small enough to be stored in set $\hat{\mathbb{Z}}_{\mathrm{bou}}$ collecting boxes containing the boundary of $\mathbb{Z}$, as shown in Line 11-12 in Algorithm~\ref{alg:Sivia}.
	\item If $[z]$ is undetermined and its width is greater than $\epsilon$, then $[z]$ should be bisected and the two newly generated boxes are stored in set $\hat{\mathbb{Z}}_\mathrm{do}$ collecting boxes needing further exploration, as shown in Line 13-15 in Algorithm~\ref{alg:Sivia}. The exploration should be recursively implemented until set $\hat{\mathbb{Z}}_\mathrm{do}$ is empty.
\end{itemize}

Two tips should be noted when SIVIA is implemented. 1) In order to guarantee that the inner and outer tests can be implemented, set $\mathbb{Y}$ and $\mathbb{Z}_\mathrm{init}$ should be approximately represented by set $\hat{\mathbb{Y}}$ and $\hat{\mathbb{Z}}_{\mathrm{init}}$ of boxes, respectively. This is without loss of generality, because any compact set can be arbitrarily approximated by unions of boxes. 2) It is important to organize the storage of the set of boxes (such as $\hat{\mathbb{Y}}, \hat{\mathbb{Z}}_{\mathrm{init}}, \hat{\mathbb{Z}}_{\mathrm{in}}, \textit{etc.}$). The first idea would be to store the set of boxes as a list. Another more efficient organization of the set of boxes should be a binary tree, which also called paving or subpaving (more details see \citealt{Jaulin:2001,Kieffer:2001}).

According to Theorem 3.1 in \cite{Jaulin:2001}, under continuity condition, the inner approximate $\hat{\mathbb{Z}}_{\mathrm{in}}$ converges to set $\mathbb{Z}$ when $\epsilon$ tends to zero. This means that SIVIA can approximate set $\mathbb{Z}$ with an arbitrary precision.

\begin{algorithm} 
	\caption{Set inversion via interval analysis} \label{alg:Sivia}
	\begin{algorithmic}[1]
		\Procedure{Sivia}{$p, \hat{\mathbb{Y}}, \hat{\mathbb{Z}}_{\mathrm{init}}, \epsilon$}
		\State $\hat{\mathbb{Z}}_{\mathrm{in}} := \emptyset, \hat{\mathbb{Z}}_{\mathrm{out}} := \emptyset, \hat{\mathbb{Z}}_\mathrm{bou} := \emptyset$
		\State $\hat{\mathbb{Z}}_{\mathrm{do}} := \hat{\mathbb{Z}}_{\mathrm{init}}$
		\While {$\hat{\mathbb{Z}}_{\mathrm{do}} \neq \emptyset$}
		\State Get a box $[z]$ from $\hat{\mathbb{Z}}_{\mathrm{do}}$
		\State Remove $[z]$ from $\hat{\mathbb{Z}}_{\mathrm{do}}$
		\If {$[p]([z]) \subset \hat{\mathbb{Y}}$}
		\State Add $[z]$ to set $\hat{\mathbb{Z}}_{\mathrm{in}}$
		\ElsIf {$[p]([z]) \cap \hat{\mathbb{Y}} = \emptyset$}
		\State Add $[z]$ to set $\hat{\mathbb{Z}}_{\mathrm{out}}$
		\ElsIf {$d([z]) < \epsilon$}
		\State Add $[z]$ to set $\hat{\mathbb{Z}}_{\mathrm{bou}}$
		\Else
		\State Bisect box $[z]$
		\State Add the two new boxes to set $\hat{\mathbb{Z}}_{\mathrm{do}}$
		\EndIf
		\EndWhile
		\State \textbf{return} $\hat{\mathbb{Z}}_{\mathrm{in}}$
		\EndProcedure
	\end{algorithmic}
\end{algorithm}

\subsection{Approximate negative-definite and invariant set by SIVIA}

In this subsection, an algorithm, estimating the negative-definite and invariant set $\mathbb{W}_{\mathrm{N\&I}}(L)$ for the given positive-definite function $L$ based on the SIVIA algorithm, is proposed. For a given Lyapunov function $L$, the negative-definite set $\mathbb{W}_N(L) \subset \mathbb{R}^n$ defined in \eqref{eq:ndd} for plant~\eqref{eq:plant} can be approximated using SIVIA algorithm. The difference $\Delta L: \mathbb{R}^{n+m} \to \mathbb{R}$ of Lyapunov function $L$ can be defined as $\Delta L (w) = L(f(w)) - L(x)$, where $w = (x;u) \in \mathbb{R}^{n+m}$. With $\Delta L$, \eqref{eq:ndd} can be rewritten as

\begin{equation}
\mathbb{W}_N(L) = \Big\{w \in \mathbb{R}^{n+m}\Big|\Delta L(w) \in [-\infty,-\alpha]\Big\}. \label{eq:ndd_Delta_L} 
\end{equation}
Comparing \eqref{eq:problem_SIVIA} and \eqref{eq:ndd_Delta_L}, an inner approximation $\hat{\mathbb{W}}_N(L)$ of negative-definite set $\mathbb{W}_N(L)$ in the constrains set $\mathbb{W}_{\mathrm{cons}}$ could be obtained using 
\begin{equation*}
\hat{\mathbb{W}}_N(L) := \mathrm{SIVIA}(\Delta L, \{[-\infty,-\alpha]\},\hat{\mathbb{W}}_{\mathrm{cons}},\epsilon),
\end{equation*}

as shown in Line 2-3 of Algorithm \ref{alg:est_W_{N&I}_(L)}. Here, the constrains set $\mathbb{W}_{\mathrm{cons}}$ is defined in Assumption~\ref{assum:cons_set} and $\hat{\mathbb{W}}_{\mathrm{cons}}$ is an approximation of $\mathbb{W}_{\mathrm{cons}}$ by a set of boxes.

According to Proposition~\ref{prop:W_N&I}, $\mathbb{W}_\mathrm{N\&I}(L) = \mathcal{I}^{\infty}(\mathbb{W}_\mathrm{N}(L))$ is an invariant subset of $\mathbb{W}_\mathrm{N}(L)$. It is needed to approximate the operation of mapping $\mathcal{I}$ using SIVIA algorithm. Comparing \eqref{eq:problem_SIVIA} and \eqref{eq:defn:I}, for a given set $\hat{\mathbb{W}}_1$ of boxes, an inner approximation $\hat{\mathbb{W}}_2$ of $\mathcal{I}(\hat{\mathbb{W}}_1)$ could be obtained using 
\begin{equation*}
\hat{\mathbb{W}}_2 := \mathrm{SIVIA}(f, \mathrm{proj}(\hat{\mathbb{W}}_1),\hat{\mathbb{W}}_1,\epsilon).
\end{equation*}
With the initial set $\hat{\mathbb{W}}_\mathrm{N}(L)$, recursively using SIVIA algorithm to approximate mapping $\mathcal{I}$, an inner approximation $\hat{\mathbb{W}}_{\mathrm{N\&I}}(L)$ of the negative-definite and invariant set $\mathbb{W}_{\mathrm{N\&I}}(L)$ defined in \eqref{eq:W_N&I} can be obtained, as shown in Line 4-10 of Algorithm \ref{alg:est_W_{N&I}_(L)}.

\begin{algorithm} 
	\caption{Inner Approximation of $\mathbb{W}_{\mathrm{N\&I}}(L)$} \label{alg:est_W_{N&I}_(L)}
	\begin{algorithmic}[1]
		\Procedure{EstNIset}{$f, L, \hat{\mathbb{W}}_{\mathrm{cons}}, \epsilon$}
		\State $\Delta L(w) := L(f(w)) - L(x)$ 
		\State $\hat{\mathbb{W}}_N(L) := \mathrm{SIVIA}(\Delta L, \{[-\infty, -\alpha]\}, \hat{\mathbb{W}}_{\mathrm{cons}}, \epsilon)$
		\State $\hat{\mathbb{W}}_1 := \hat{\mathbb{W}}_N(L)$
		\State $\hat{\mathbb{W}}_2 := \emptyset$
		\While{$\hat{\mathbb{W}}_1 \neq \hat{\mathbb{W}}_2$}
		\State $\hat{\mathbb{W}}_2 := \hat{\mathbb{W}}_1$
		\State $\hat{\mathbb{W}}_1 := \mathrm{SIVIA}(f, \mathrm{proj}(\hat{\mathbb{W}}_2), \hat{\mathbb{W}}_2, \epsilon)$	
		\EndWhile
		\State $\hat{\mathbb{W}}_{\mathrm{N\&I}}(L) := \hat{\mathbb{W}}_1$
		\State \textbf{return} $\hat{\mathbb{W}}_{\mathrm{N\&I}}(L)$
		\EndProcedure
	\end{algorithmic}
\end{algorithm} 

\subsection{Estimate closed-loop DOA and design controller} \label{sec:estimate_DOA}

With the estimate $\hat{\mathbb{W}}_{\mathrm{N\&I}}(L)$ of the negative-definite and invariant set $\mathbb{W}_{\mathrm{N\&I}}(L)$, a method of estimating the closed-loop DOA and designing the controller is proposed in this subsection. Since $L(f(0,0)) - L(0) = 0$, the origin $0 \in \mathbb{R}^{n+m}$ is in the boundary of $\mathbb{W}_{\mathrm{N\&I}}(L)$. Hence, there is no box belonging to the inner approximation $\hat{\mathbb{W}}_{\mathrm{N\&I}}(L)$ nearby the origin $0 \in \mathbb{R}^{n+m}$ and there is a small neighborhood $\mathbb{X}_0$ of the origin $0 \in \mathbb{R}^{n}$ that is not contained by $\mathrm{proj}(\hat{\mathbb{W}}_{\mathrm{N\&I}}(L))$. The size of the neighborhood $\mathbb{X}_0$ depends on parameter $\epsilon$ in Algorithm~\ref{alg:est_W_{N&I}_(L)}. Since we suppose that the linearization of plant \eqref{eq:plant} is controllable at the origin, there must exist a linear controller which is able to stabilize all state in $\mathbb{X}_0$ when the size of $\mathbb{X}_0$ is small enough. As a result, the estimate of the closed-loop DOA should be $\mathrm{proj}(\hat{\mathbb{W}}_{\mathrm{N\&I}}(L)) \cup \mathbb{X}_0$ and the corresponding state feedback controller is 
\begin{equation} \label{eq:controller}
\mu(x) = \left\{ 
\begin{array}{ll}
Kx, & \mathrm{if \ } x \in \mathbb{X}_0 \\
\tilde{\mu}(x), & \mathrm{if \ } x \in \mathrm{proj}(\hat{\mathbb{W}}_{\mathrm{N\&I}}(L))
\end{array}
\right.,
\end{equation}
where $K \in \mathbb{R}^{m \times n}$ is obtained by the linear controller design method according to the linearization of plant \eqref{eq:plant} and $\tilde{\mu}: \mathbb{R}^n \to \mathbb{R}^m$ satisfies $\forall x \in \mathrm{proj}(\hat{\mathbb{W}}_{\mathrm{N\&I}}(L)), (x;\tilde{\mu}(x)) \in \hat{\mathbb{W}}_{\mathrm{N\&I}}(L)$. A simple way to find such a $\tilde{\mu}$ is that, first, select a training set belonging to $\hat{\mathbb{W}}_{\mathrm{N\&I}}(L)$; then, obtain $\tilde{\mu}$ with a function estimation method, such as interpolation, Gaussian processes regression and so on. When the trend of the training data points is smooth enough, it can be guaranteed that $\tilde{\mu}$ obtained from the function estimator satisfies $\forall x \in \mathrm{proj}(\hat{\mathbb{W}}_{\mathrm{N\&I}}(L)), (x;\tilde{\mu}(x)) \in \hat{\mathbb{W}}_{\mathrm{N\&I}}(L)$. 

\subsection{Enlargement of closed-loop DOA by selecting Lyapunov function}

It is observed that, for different Lyapunov functions, the negative-definite and invariant sets $\mathbb{W}_{\mathrm{N\&I}}(L)$ are totally different. Hence, if a positive-definite function set rather than a Lyapunov function is given, a significantly large estimate of the closed-loop DOA may be obtained by selecting an appropriate Lyapunov function from the positive-definite function set. Based on this idea, the following optimization problem is formulated.
\begin{equation}
\max_{L \in \mathfrak{L}_{n,2d}} \mathfrak{m}(\mathrm{proj}(\hat{\mathbb{W}}_{\mathrm{N\&I}}(L))), \label{eq:optim}
\end{equation}

where $\mathfrak{m}(\mathrm{proj}(\hat{\mathbb{W}}_{\mathrm{N\&I}}(L)))$ denotes the Lebesgue measure of $\mathrm{proj}(\hat{\mathbb{W}}_{\mathrm{N\&I}}(L))$ (in Euclidean space, it is the volume) and $\mathfrak{L}_{n,2d}$ is a subset of all sum-of-square polynomials \citep{Packard:2010_18} in $n$ variables with degree $\leq 2d$, which is defined as
\begin{equation}
\mathfrak{L}_{n,2d} = \Big\{L \in \mathfrak{R}_{n,2d} \Big| L(x) = s^T_d(x)P^TPs_d(x), x \in \mathbb{R}^n\Big\}, \nonumber
\end{equation}
where $\mathfrak{R}_{n,2d}$ denotes the set of all polynomials in $n$ variables with degree $\leq 2d$, $P \in \mathbb{R}^{r \times r}$ has full rank,

\begin{equation}
s_d(x) = (x_{(1)};\cdots;x_{(n)};x_{(1)} x_{(2)};\cdots;x^d_{(n)}) \in \mathbb{R}^r, \nonumber
\end{equation}
and $r = \left(
\begin{smallmatrix}
\scriptscriptstyle n+d \\
\scriptscriptstyle d
\end{smallmatrix} \right) - 1$.

Although the estimate of the closed-loop DOA is $\mathrm{proj}(\hat{\mathbb{W}}_{\mathrm{N\&I}}(L))$ $\cup \mathbb{X}_0$ from Section \ref{sec:estimate_DOA}, the objective of optimization problem~\eqref{eq:optim} is focused on enlarging the volume of $\mathrm{proj}(\hat{\mathbb{W}}_{\mathrm{N\&I}}(L))$ due to that $\mathbb{X}_0$ is much small. Two points should be noted about optimization problem~\eqref{eq:optim}. First, for the given positive-definite function $L$, the volume of $\mathrm{proj}(\hat{\mathbb{W}}_{\mathrm{N\&I}}(L))$ is easy to calculated, since $\mathrm{proj}(\hat{\mathbb{W}}_{\mathrm{N\&I}}(L))$ consists of boxes. Second, the positive-definite function $L$ is selected from a parameterized polynomial function set with the parameters $P \in \mathbb{R}^{r \times r}$. We defined function $m: \mathbb{R}^{r \times r} \to \mathbb{R}$ as $m(P) = \mathfrak{m}(\mathrm{proj}(\hat{\mathbb{W}}_{\mathrm{N\&I}}(L)))$, where $L(x) = s^T_d(x)P^TPs_d(x)$ and $\hat{\mathbb{W}}_{\mathrm{N\&I}}(L)$ can be obtained by Algorithm~\ref{alg:est_W_{N&I}_(L)}. With the function $m(P)$, the optimization problem~\eqref{eq:optim} can be equivalently rewritten as
\begin{equation}
\max_{P \in \mathbb{R}^{r \times r}} m(P). \label{eq:optim_P}
\end{equation}
The analytical expression of $m(P)$ is hard to be derived, but it is easy to evaluate $m(P)$ for a given $P$. Hence, classic optimization methods, e.g., gradient descent method, cannot be used to solve the optimization problem~\eqref{eq:optim_P}. However, meta-heuristic optimization methods can be used to solve the optimization problem~\eqref{eq:optim_P}, whose advantage is that the function to be optimized is only required to be evaluable. Popular meta-heuristic optimizers for real-valued search-spaces include particle swarm optimization, differential evolution and evolution strategies. There are lots of literatures about meta-heuristic optimizers, so we omit the introduction about them in this paper.

\section{Simulation}

Consider the plant 
\begin{equation*}
x(k+1) = -\sin(2x(k)) - x(k)u(k) - 0.2x(k) - u^2(k) + u(k),
\end{equation*}
where $x(k) \in \mathbb{R}$ and $u(k) \in \mathbb{R}$. The constrains set in Assumption~\ref{assum:cons_set} is $\mathbb{W}_{\mathrm{cons}} = [-2,2] \times [-2,2] \subset \mathbb{R}^2$.

\subsection{Stabilization with closed-loop DOA estimation for given Lyapunov function}

\begin{figure}
	\begin{center}
		\includegraphics[width=0.236\textwidth]{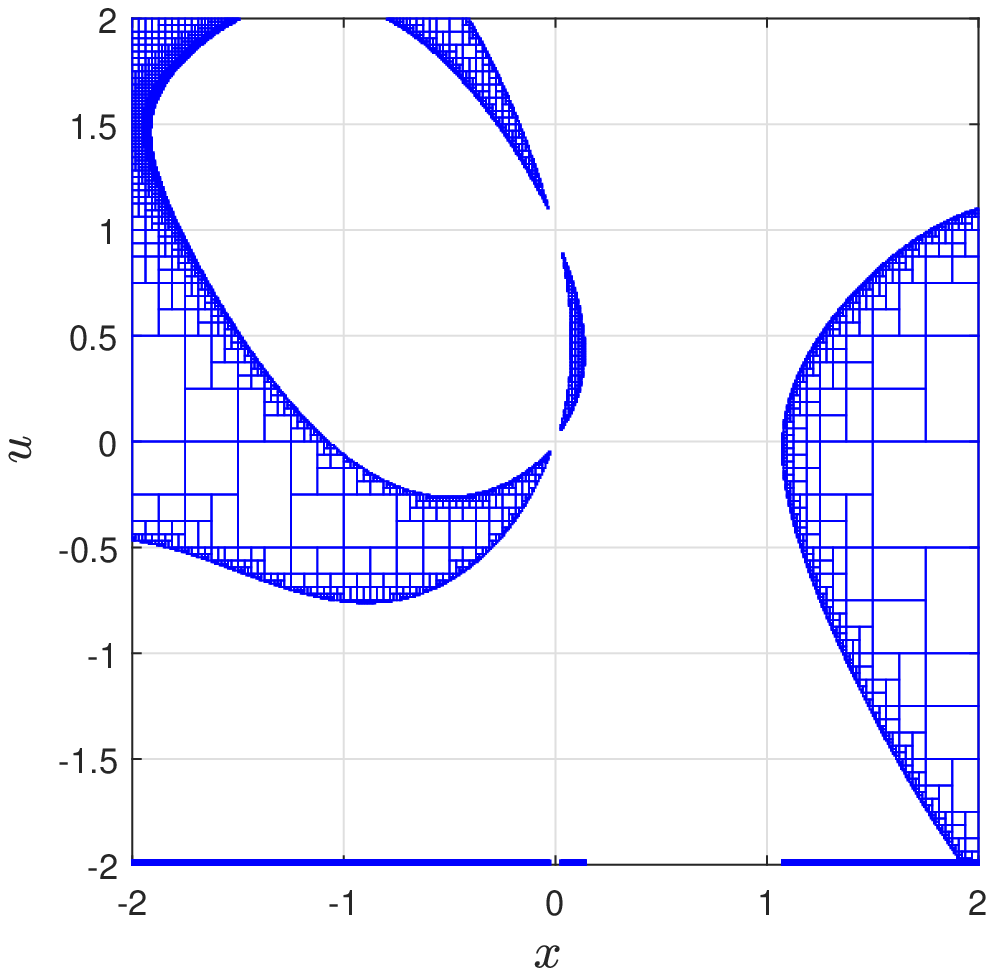}
		\includegraphics[width=0.236\textwidth]{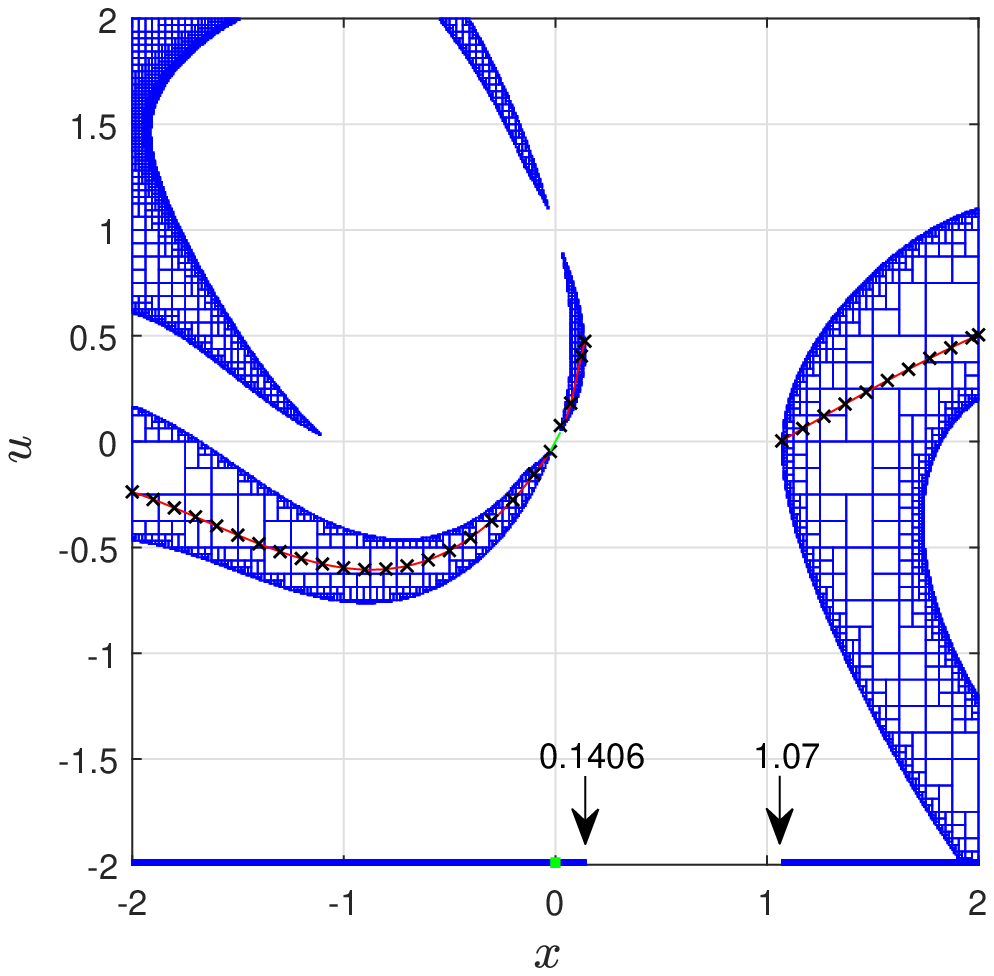} \\
		\parbox[c]{0.236\textwidth}{\footnotesize \centering (a)}
		\parbox[c]{0.236\textwidth}{\footnotesize \centering (b)}\\
		\includegraphics[width=0.236\textwidth]{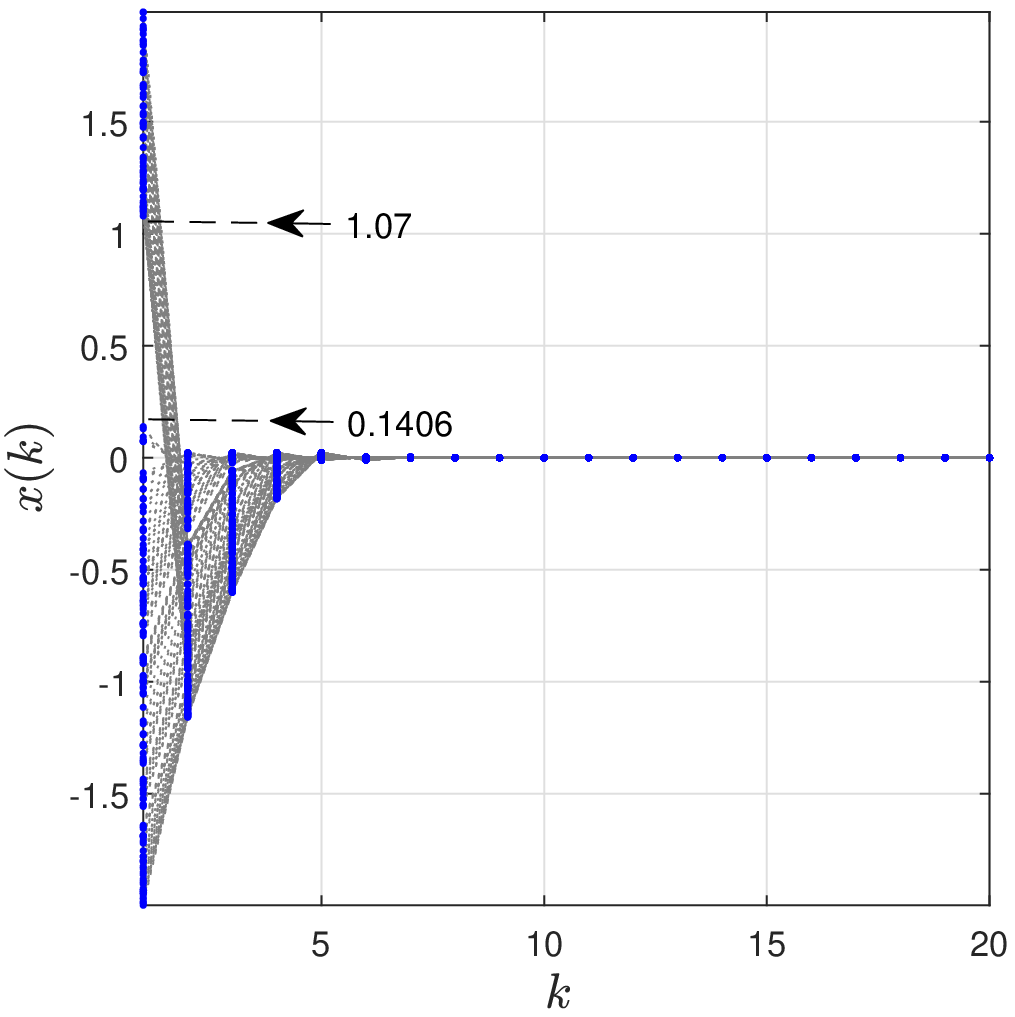}
		\includegraphics[width=0.236\textwidth]{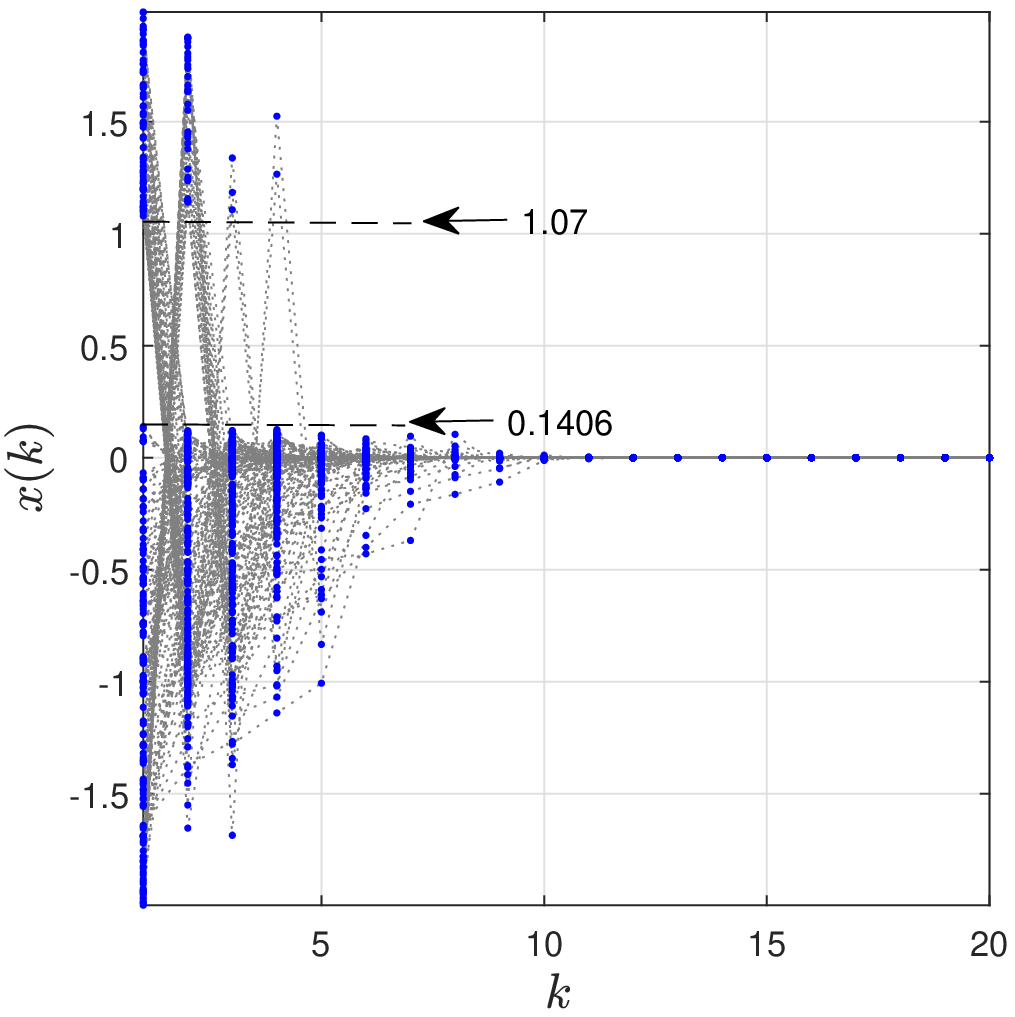}\\
		\parbox[c]{0.236\textwidth}{\footnotesize \centering (c)}
		\parbox[c]{0.236\textwidth}{\footnotesize \centering (d)}
		\caption{(a) Inner approximation $\hat{\mathbb{W}}_{\mathrm{N}}(L)$ of the negative-definite set $\mathbb{W}_{\mathrm{N}}(L)$. (b) Inner approximation $\hat{\mathbb{W}}_{\mathrm{N}\&\mathrm{I}}(L)$ of the negative-definite and invariant set $\mathbb{W}_{\mathrm{N}\&\mathrm{I}}(L)$, estimate of the closed-loop DOA $\mathrm{proj}(\hat{\mathbb{W}}_{\mathrm{N\&I}}(L)) \cup \mathbb{X}_0$ and controller $u = \mu(x)$. (c) Trajectories of the closed-loop with the controller $u = \mu(x)$. (d) Trajectories of the closed-loop with control inputs drawn from the uniform distribution on $\mathbb{U}(x)$.} \label{fig:exmp:x^2}
	\end{center}
\end{figure}

The given Lyapunov function is selected as $L(x) = x^2$ and the constant $\alpha$ in \eqref{eq:ndd_Delta_L} is selected as $\alpha = 10^{-15}$. In Algorithm~\ref{alg:est_W_{N&I}_(L)}, the parameter $\epsilon$ is selected as $\epsilon = 0.01$. The inner approximation $\hat{\mathbb{W}}_{\mathrm{N}}(L)$ of the negative-definite set $\mathbb{W}_{\mathrm{N}}(L)$ is shown in Fig.~\ref{fig:exmp:x^2} (a) denoted by blue boxes. $\mathrm{proj}(\hat{\mathbb{W}}_{\mathrm{N}}(L)) = [-2,-0.02344] \cup [0.02344,0.1406] \cup [1.07,2]$ is also shown in  Fig.~\ref{fig:exmp:x^2} (a) denoted by blue line segments in $x$-axis. The inner approximation $\hat{\mathbb{W}}_{\mathrm{N}\&\mathrm{I}}(L)$ of the negative-definite and invariant set $\mathbb{W}_{\mathrm{N}\&\mathrm{I}}(L)$ is shown in Fig.~\ref{fig:exmp:x^2} (b) denoted by blue boxes. $\mathrm{proj}(\hat{\mathbb{W}}_{\mathrm{N}\&\mathrm{I}}(L)) = [-2,-0.02344] \cup [0.02344,0.1406] \cup [1.07,2]$ is also shown in  Fig.~\ref{fig:exmp:x^2} (b) denoted by blue line segments in $x$-axis. The small neighborhood $\mathbb{X}_0$ in \eqref{eq:controller} of the origin is $[-0.02344,0.02344]$ as shown in Fig.~\ref{fig:exmp:x^2} (b) denoted by the green line segment in $x$-axis. Hence, the estimate of the closed-loop DOA is $\mathrm{proj}(\hat{\mathbb{W}}_{\mathrm{N\&I}}(L)) \cup \mathbb{X}_0 = [-2,0.1406] \cup [1.07,2]$. When the invariant of the estimate of the closed-loop DOA is guaranteed by the level-set of $L(x) = x^2$, \textit{e.g.}, the method in \cite{Li:2014_79}, the result is $[-0.1406,0.1406]$.

The linear controller in \eqref{eq:controller} is $u = 1.8649x$ denoted by the green straight line through the origin in Fig.~\ref{fig:exmp:x^2} (b). In order to find the nonlinear controller $\tilde{\mu}$ in \eqref{eq:controller}, we select a training data set denoted by black 'x's in Fig.~\ref{fig:exmp:x^2} (b). Then, $\tilde{\mu}$ is obtained using Gaussian processes regression, denoted by red line in Fig.~\ref{fig:exmp:x^2} (b). Fig.~\ref{fig:exmp:x^2} (c) shows 200 state trajectories of the closed-loop, whose initial states are drawn from the uniform distribution on $[-2,0.1406] \cup [1.07,2]$. We see that all state trajectories converge to the origin.

To verify whether all controllers belonging to $\hat{\mathbb{W}}_{\mathrm{N}\&\mathrm{I}}(L)$ can stabilize the plant, Fig.~\ref{fig:exmp:x^2} (d) also shows 200 state trajectories of the closed-loop, whose initial states are drawn from the uniform distribution on $[-2,0.1406] \cup [1.07,2]$. Here, the linear controller is still $u = 1.8649x$, while the output of the nonlinear controller is drawn from the uniform distribution on $\mathbb{U}(x) \subset \mathbb{R}$. For the given $x \in \mathbb{R}$, $\mathbb{U}(x)$ is defined as $\mathbb{U}(x) = \big\{u \in \mathbb{R} | (x;u) \in \hat{\mathbb{W}}_{\mathrm{N}\&\mathrm{I}}(L) \big\}$. We see that all state trajectories converge to the origin.

State trajectories shown in Fig.~\ref{fig:exmp:x^2} (c) and (d) also verify that $\hat{\mathbb{W}}_{\mathrm{N}\&\mathrm{I}}(L) \cup \{(x;u)|u = 1.8649x, x \in \mathbb{X}_0 \}$ is invariant for the plant. We know that $\mathrm{proj}(\hat{\mathbb{W}}_{\mathrm{N}\&\mathrm{I}}(L) \cup \{(x;u)|u = 1.8649x, x \in \mathbb{X}_0 \}) = \mathrm{proj}(\hat{\mathbb{W}}_{\mathrm{N}\&\mathrm{I}}(L)) \cup \mathbb{X}_0 = [-2,0.1406] \cup [1.07,2]$. From the figures, we see that there is no state in $(0.1406,1.07) \subset \mathbb{R}$.

\subsection{Stabilization with closed-loop DOA enlargement}

\begin{figure}
	\begin{center}
		\includegraphics[width=0.236\textwidth]{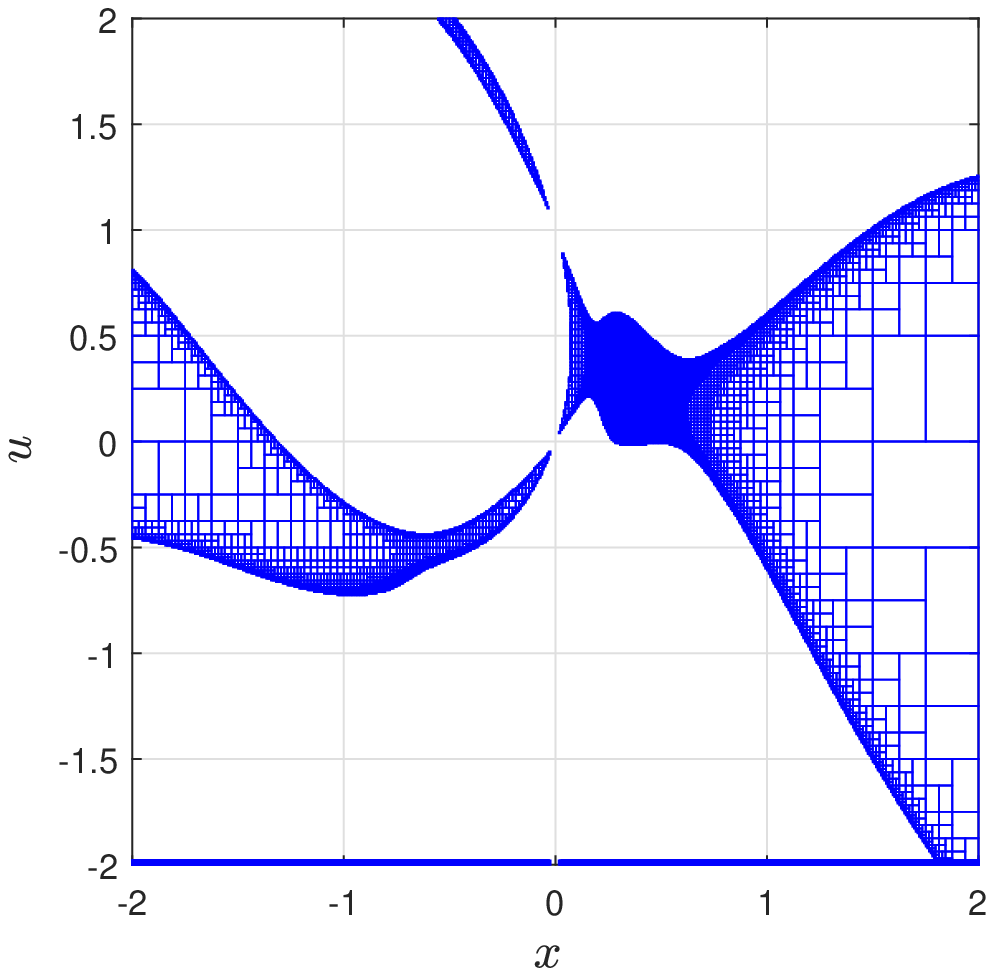}
		\includegraphics[width=0.236\textwidth]{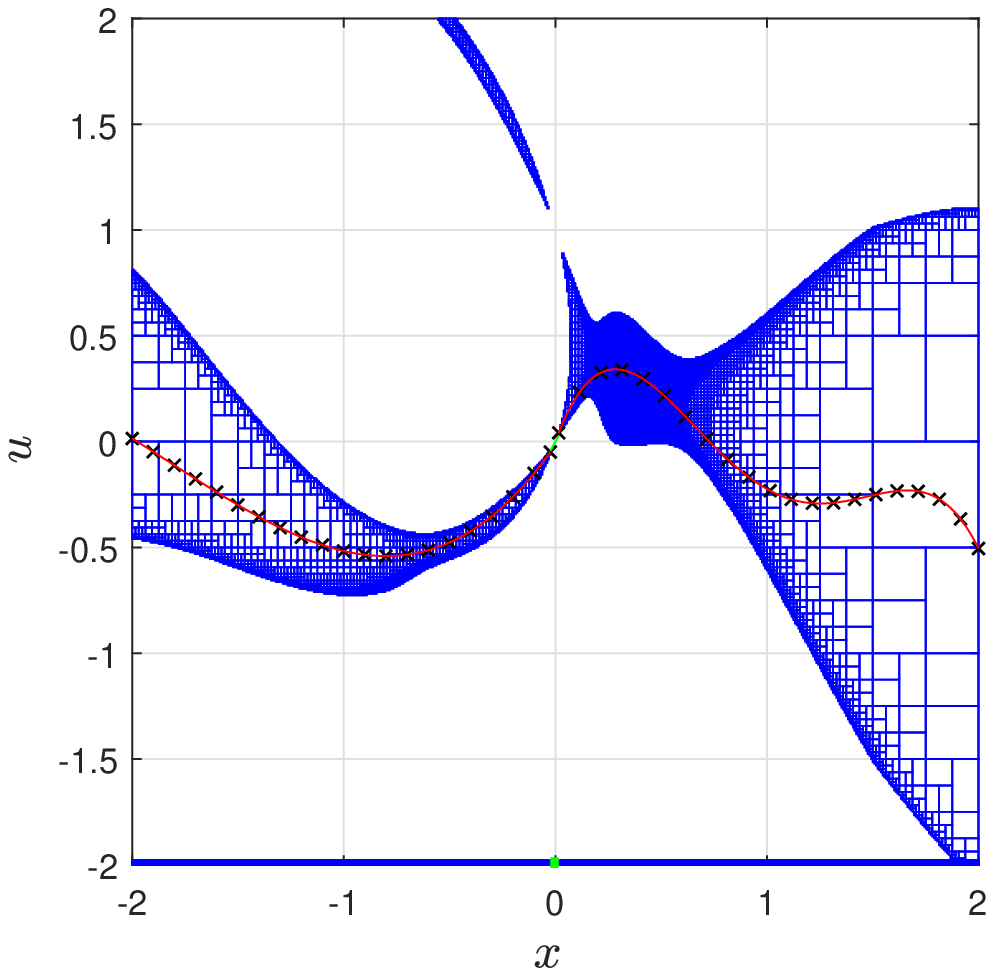} \\
		\parbox[c]{0.236\textwidth}{\footnotesize \centering (a)}
		\parbox[c]{0.236\textwidth}{\footnotesize \centering (b)}\\
		\includegraphics[width=0.236\textwidth]{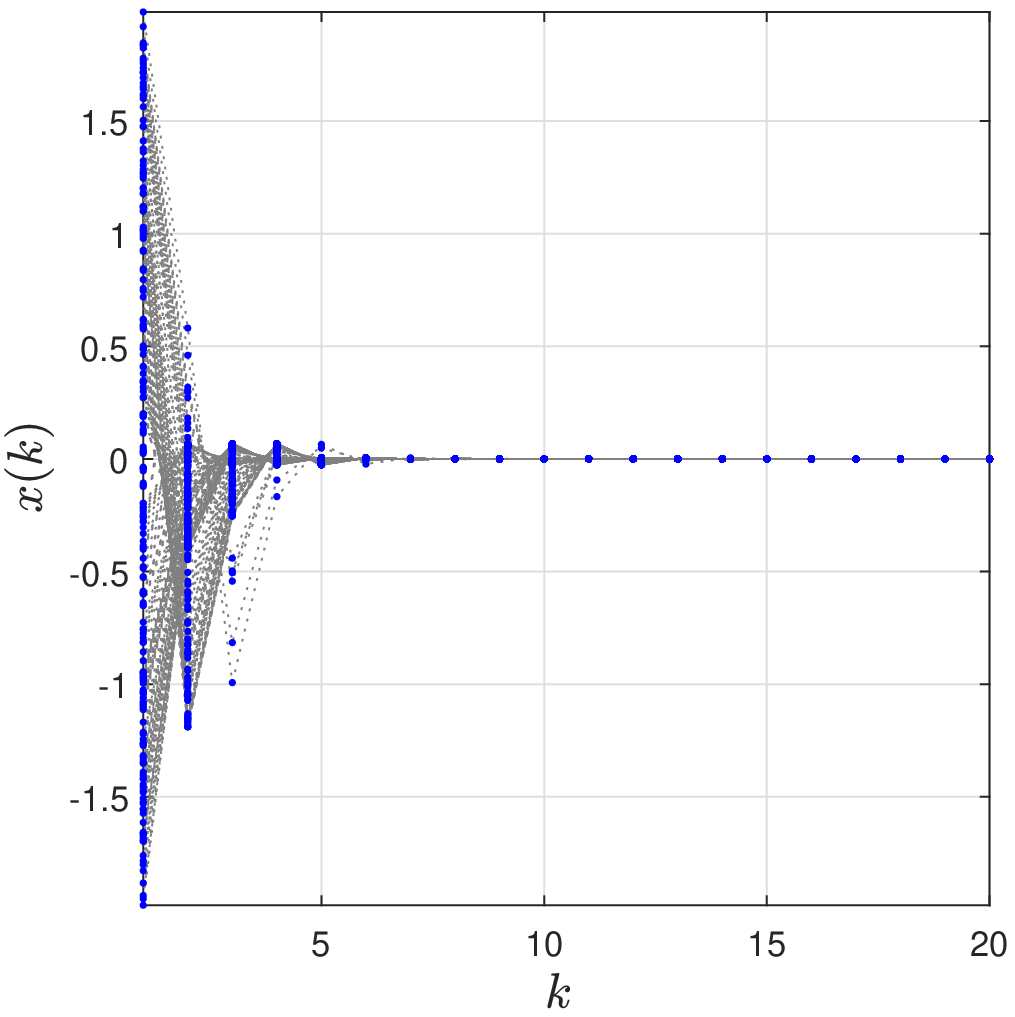}
		\includegraphics[width=0.236\textwidth]{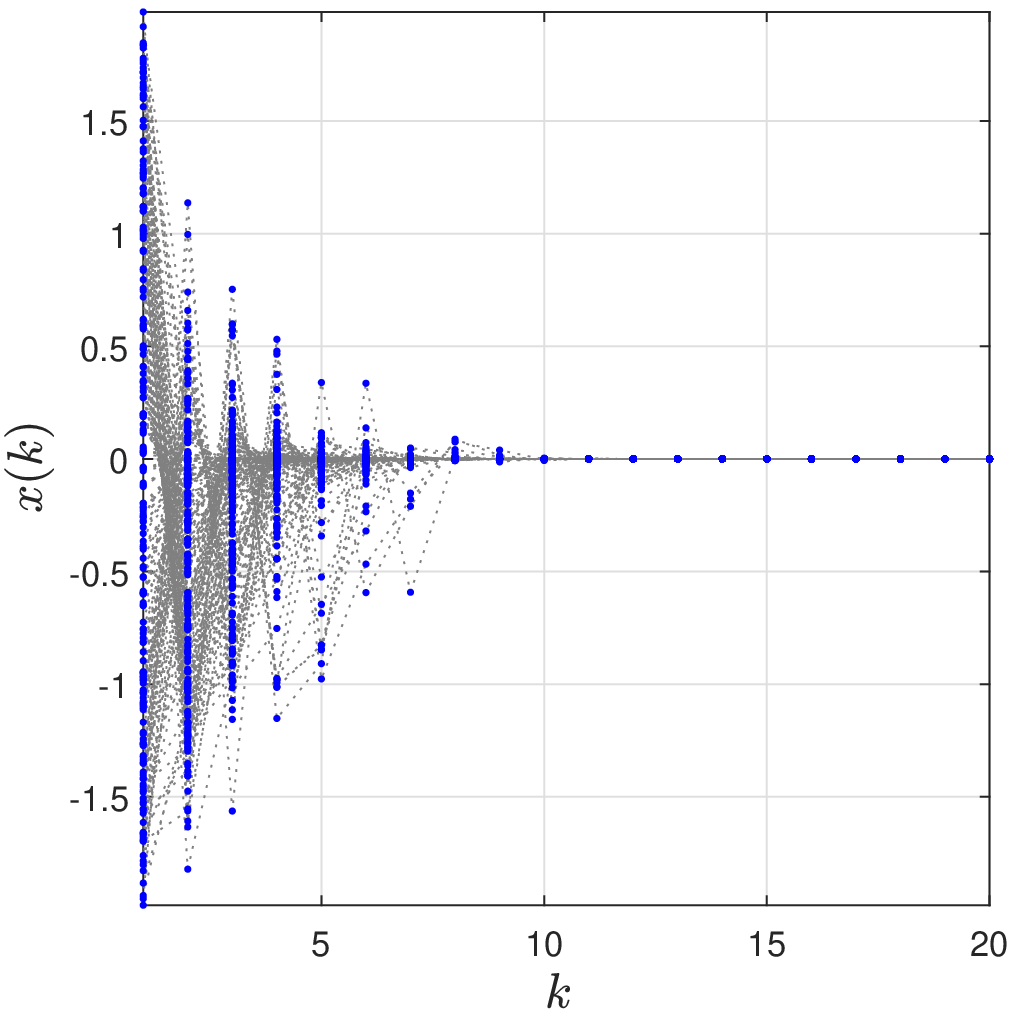}\\
		\parbox[c]{0.236\textwidth}{\footnotesize \centering (c)}
		\parbox[c]{0.236\textwidth}{\footnotesize \centering (d)}
		\caption{(a) Inner approximation $\hat{\mathbb{W}}_{\mathrm{N}}(L^\ast)$ of the negative-definite set $\mathbb{W}_{\mathrm{N}}(L^\ast)$. (b) Inner approximation $\hat{\mathbb{W}}_{\mathrm{N}\&\mathrm{I}}(L^\ast)$ of the negative-definite and invariant set $\mathbb{W}_{\mathrm{N}\&\mathrm{I}}(L^\ast)$, estimate of the closed-loop DOA $\mathrm{proj}(\hat{\mathbb{W}}_{\mathrm{N\&I}}(L^\ast)) \cup \mathbb{X}_0$ and controller $u = \mu(x)$. (c) Trajectories of the closed-loop with the controller $u = \mu(x)$. (d) Trajectories of the closed-loop with control inputs drawn from the uniform distribution on $\mathbb{U}(x)$.} \label{fig:exmp:opt}
	\end{center}
\end{figure}

The positive-definite function set in optimization problem \eqref{eq:optim} is selected as
\begin{equation}
\mathfrak{L}_{1,4} = = \Big\{L \in \mathfrak{R}_{1,4} \Big| L(x) = (x;x^2)^TP^TP(x;x^2), x \in \mathbb{R}\Big\} \nonumber
\end{equation}
with the parameters $P \in \mathbb{R}^{2 \times 2}$. The optimization problem \eqref{eq:optim_P} is solved through the particle swarm optimization method and the solution $L^\ast(x) = 2.4468x^2 + 3.4186x^3 + 1.4524x^4$ is obtained.

The inner approximation $\hat{\mathbb{W}}_{\mathrm{N}}(L^\ast)$ of the negative-definite set $\mathbb{W}_{\mathrm{N}}(L^\ast)$ is shown in Fig.~\ref{fig:exmp:opt} (a) denoted by blue boxes. The inner approximation $\hat{\mathbb{W}}_{\mathrm{N}\&\mathrm{I}}(L^\ast)$ of the negative-definite and invariant set $\mathbb{W}_{\mathrm{N}\&\mathrm{I}}(L^\ast)$ is shown in Fig.~\ref{fig:exmp:opt} (b) denoted by blue boxes. $\mathrm{proj}(\hat{\mathbb{W}}_{\mathrm{N}\&\mathrm{I}}(L^\ast)) = [-2,-0.02344] \cup [0.01563,2]$ is also shown in  Fig.~\ref{fig:exmp:opt} (b) denoted by blue line segments in $x$-axis. The small neighborhood $\mathbb{X}_0$ in \eqref{eq:controller} of the origin is $[-0.02344,0.01563]$ as shown in Fig.~\ref{fig:exmp:opt} (b) denoted by the green line segment in $x$-axis. Hence, the estimate of the closed-loop DOA is $\mathrm{proj}(\hat{\mathbb{W}}_{\mathrm{N\&I}}(L^\ast)) \cup \mathbb{X}_0 = [-2,2]$. When the invariant of the estimate of the closed-loop DOA is guaranteed by the level-set of $L^\ast(x)$, \textit{e.g.}, the method in \cite{Li:2014_79}, the result is $[-2,0.9145]$.

The linear controller in \eqref{eq:controller} is $u = 1.8649x$ denoted by the green straight line through the origin in Fig.~\ref{fig:exmp:opt} (b). In order to find the nonlinear controller $\tilde{\mu}$ in \eqref{eq:controller}, we select a training data set denoted by black 'x's in Fig.~\ref{fig:exmp:opt} (b). Then, $\tilde{\mu}$ is obtained using Gaussian processes regression, denoted by red line in Fig.~\ref{fig:exmp:opt} (b). Fig.~\ref{fig:exmp:opt} (c) shows 200 state trajectories of the closed-loop, whose initial states are drawn from the uniform distribution on $[-2,2]$. We see that all state trajectories converge to the origin.

To verify whether all controllers belonging to $\hat{\mathbb{W}}_{\mathrm{N}\&\mathrm{I}}(L^\ast)$ can stabilize the plant, Fig.~\ref{fig:exmp:opt} (d) also shows 200 state trajectories of the closed-loop, whose initial states are drawn from the uniform distribution on $[-2,2]$. Here, the linear controller is still $u = 1.8649x$, while the output of the nonlinear controller is drawn from the uniform distribution on $\mathbb{U}(x) \subset \mathbb{R}$. For the given $x \in \mathbb{R}$, $\mathbb{U}(x)$ is defined as $\mathbb{U}(x) = \big\{u \in \mathbb{R} | (x;u) \in \hat{\mathbb{W}}_{\mathrm{N}\&\mathrm{I}}(L^\ast) \big\}$. We see that all state trajectories converge to the origin.

\section{Conclusion} \label{sec:conclusion}

For general nonlinear plants, due to the difficulty to achieve the global stabilization, the DOA of the closed-loop requires extensive investigation. In this paper, considering discrete-time general nonlinear plants, the stabilization with closed-loop DOA enlargement is solved. First, a sufficient condition for asymptotic stabilization and estimation of the closed-loop DOA is given, which shows that, for a given Lyapunov function, the negative-definite and invariant set in the state-control space is a stabilizing controller set and its projection along the control space to the state space can be an estimate of the closed-loop DOA. Then, an algorithm is proposed to approximate the negative-definite and invariant set for the given Lyapunov function, in which SIVIA algorithm is used to find inner approximations of interesting sets as precise as desired. Finally, a solvable optimization problem is formulated to enlarge the estimate of the closed-loop DOA by selecting an appropriate Lyapunov function from a positive-definite function set. Some examples are also includes to verify the effectiveness of the proposed method.

\section{Acknowledgments}

The authors would like to thank Dr. Yinan Li for helpful discussions.

\end{document}